\documentclass{amsproc}

\usepackage[margin=1in]{geometry}
\usepackage{amsfonts}
\usepackage{amssymb}
\usepackage[dvips]{graphics}
\usepackage{epsfig}
\usepackage{euscript}
\usepackage{color}
\usepackage{bbm}
\usepackage[colorlinks ,pdfstartview=FitB]{hyperref}
\usepackage{pgf,tikz}
\usepackage{mathrsfs}
\usetikzlibrary{arrows}
\hypersetup{allcolors=blue}

\newtheorem{thm}{Theorem}[section]

\newtheorem{proposition}[thm]{Proposition}
\theoremstyle{definition}

\numberwithin{thm}{section}

\newcommand{\N}{{\mathord{\mathbb N}}}

\newcommand{\cH}{\mathcal{H}}

\newcommand{\cN}{\mathcal{N}}

\newcommand{\C}{{\mathord{\mathbb C}}}
\newcommand{\Z}{{\mathord{\mathbb Z}}}

\newcommand{\cG}{\mathcal{G}}
\newcommand{\cV}{\mathcal{V}}
\newcommand{\cL}{\mathcal{L}}
\newcommand{\cE}{\mathcal{E}}

\def\idty{{\mathchoice {\mathrm{1\mskip-4mu l}} {\mathrm{1\mskip-4mu l}} %
{\mathrm{1\mskip-4.5mu l}} {\mathrm{1\mskip-5mu l}}}}

\DeclareMathOperator{\spann}{span}
\DeclareMathOperator{\supp}{supp}

\newtheorem{theorem}{Theorem}[section]
\newtheorem{lemma}[theorem]{Lemma}

\theoremstyle{definition}

\newtheorem{example}[theorem]{Example}

\theoremstyle{remark}
\newtheorem{remark}[theorem]{Remark}

\numberwithin{equation}{section}

\begin{document}

\title[Higher Spin XXZ Systems on General Graphs]{A Schr\"odinger Operator Approach to Higher Spin XXZ Systems on General Graphs}


\author{Christoph Fischbacher}
\address{Department of Mathematics, University of Alabama at Birmingham, Birmingham, AL 35294, USA}
\email{cfischb@uab.edu}
\thanks{}

\subjclass[2000]{Primary 82B20}

\date{}

\begin{abstract} We consider the spin--$J$ XXZ--Hamiltonian on general graphs $\cG$ and show its equivalence to a direct sum of discrete many--particle Schr\"odinger type operators on what we call ``$N$-particle graphs with maximal local occupation number $M$", where the kinetic term is described by a weighted Laplacian. Generalizing previous results for the spin--$1/2$ case, we give sufficient conditions for the existence of spectral gaps above the low--lying droplet band when the underlying graph $\cG$ is (i) the chain and (ii) a strip of width $L$.  
\end{abstract}

\maketitle
\section{Introduction} 
Recently, there has been a significant increase of interest towards the Heisenberg spin--$1/2$ XXZ model as it is one of the first models for which phenomena associated with Many Body Localization (MBL) were rigorously proven \cite{BW1, BW2, EKS1, EKS2}. One of the many interesting open questions related to MBL for the XXZ model is whether previous results can be generalized to the case of higher spins (cf. \cite{StolzProc}). The purpose of this paper is to lay a foundation for further investigations of the XXZ model on this direction.
  
The Heisenberg spin--$1/2$ XXZ model has been subject to numerous scientific publications. In particular, it is well--known that the XXZ model defined on an infinite chain is exactly solvable with the help of the Bethe ansatz \cite{BabbittThomas1977,BabbittThomas1978,BabbittThomas1979,BabbittGutkin1990,Bethe,Borodinetal,Gutkin1993,Gutkin2000,Thomas1977}. 
For more results on the low--lying spectrum of the XXZ chain, in particular questions about the ground state gap and the so called droplet spectrum, we refer the interested reader to \cite{FS13,KN1,KN2,NS,NSS,Starr} for the spin--$1/2$ case and to \cite{KNS,MNSS} for the case of higher spins. 

In a recent paper \cite{FS18}, the Heisenberg XXZ spin-$1/2$ model defined on general graphs was investigated. Using the Hamiltonian's conservation of total magnetization, the underlying Hilbert space was decomposed into subspaces of fixed total number of down-spins, which are then viewed as particles embedded in a ``vacuum" of up-spins. It was shown that the XXZ Hamiltonian restricted to such a subspace of fixed magnetization/particle number is unitarily equivalent to a discrete Schr\"odinger operator of hard--core bosons with an interaction potential that energetically favors states that minimize their edge surface or in other words: form droplets. 
This approach has proven itself to be useful as it makes the model accessible to Schr\"odinger operator methods. In particular, while not being exactly solvable with the Bethe ansatz anymore, it was still possible to show that for quasi one--dimensional graphs (e.g.\ on strips), the notion of droplet spectrum that is uniform in the total magnetization/particle number still persists -- a result which we will generalize to the case of higher spins.

We will proceed as follows:

In Section \ref{sec:XXZ}, we introduce the spin--$J$ XXZ Hamiltonian $H_\cG$ on finite graphs $\cG$. Using the operator's conservation of total magnetization, we decompose the Hilbert space in subspaces of fixed total magnetization and study the Hamiltonian restricted to these subspaces. An analogous approach is to view the conserved total magnetization as a conserved total number of particles $N$. In this framework, a spin state at a site of the graph corresponds to a certain number of particles occupying that site. As the state space of a single spin is described by a finite--dimensional Hilbert space, the numbers of particles that can occupy a single site is limited. We then construct what we call ``$N$--particle graphs with maximal local occupation number $M$", which we denote by $\cG^{M,N}$ and show in Proposition \ref{prop:equivalence}, that the spin--$J$ Hamiltonian restricted to spaces of fixed total magnetization is unitarily equivalent to a discrete Schr\"odinger operator of the form $-\frac{1}{2\Delta}A_N+V_N$, where $A_N$ is a weighted adjacency operator on $\cG^{M,N}$ and $V_N$ is a multiplication operator by an interaction potential. Compared to the case $J=1/2$, which was studied in \cite{FS18}, and where it was possible to relate the value of the interaction potential to the edge--surface of spin configurations, the interaction potential $V_N$ takes a significantly more complicated form for higher values of $J$. Nevertheless, as in the spin--$1/2$ case, this equivalence makes it still accessible to Schr\"odinger operator methods.
 
In Section \ref{sec:gaps}, we adapt a result from \cite{FS18}, where a rather general sufficient condition for the existence of spectral gaps for discrete Schr\"odinger operators on general graphs has been established. As this result only allowed Schr\"odinger operators with non--weighted adjacency operators to be considered, we present the necessary modifications that have to be made that allow us to cover the case of weighted adjacency operators.

Finally, in Section \ref{sec:droplet}, we discuss the existence of gaps in the low--lying spectrum of $H_\cG$. We apply the results from the previous sections in order to show sufficient conditions on the anisotropy parameter $\Delta$ that ensure the existence of a gap above what we call the droplet spectrum. The cases we will study are the spin--$J$ model defined on the chain and on strips of width $L$. 
\vspace{0.3cm}

{\bf Acknowledgements:} This article contains a generalized version of results presented by the author during the Arizona School of Analysis and Mathematical Physics at the University of Arizona in March 2018. I would like to thank its organizers Houssam Abdul--Rahman, Robert Sims and Amanda Young for inviting me to attend and for providing me with the opportunity to give a talk there.
I am also very grateful to G\"unter Stolz for the results we have obtained together for the spin--$1/2$ case and for all the useful and insightful discussions. I would also like to thank Luis Manuel Rivera for pointing out some useful graph--theoretic references.

\section{The Spin-$J$ XXZ model on general graphs} \label{sec:XXZ}
\subsection{Spin--$J$ at single sites}
For any spin number $J\in\frac{1}{2}\N=\{1/2,1,3/2,2,\dots\}$ let $M+1:=2J+1$ be the dimension of the \emph{single site} Hilbert space $\C^{M+1}$. Moreover, let $\{\delta_i\}_{i=-J}^{J}$ denote the canonical orthonormal basis of $\C^{M+1}$, i.\ e.\ $$(\delta_i)_j=\begin{cases} 1 \quad\mbox{if} \quad i=j\\ 0 \quad\mbox{else}\end{cases}.$$ For convenience, we introduce the Pauli spin-$J$ matrices $S^1, S^2$ and $S^3$ in the following way (cf.\ \cite{MNSS}). Firstly, let the spin--lowering operator $S^-$ and the spin raising operator $S^+$ be given by
\begin{align}
S^-\delta_j:=\begin{cases}\sqrt{J(J+1)-j(j-1)}\delta_{j-1}\quad &j\in\{-J+1,-J+2,\dots,J-1,J\}\\
0\quad &\mbox{else.}
\end{cases}\\
S^+\delta_j:=\begin{cases}\sqrt{J(J+1)-j(j+1)}\delta_{j+1}\quad &j\in\{-J,-J+1,\dots,J-2,J-1\}\\
0\quad &\mbox{else.}
\end{cases}
\end{align}
The operators $S^1$ and $S^2$ are then given by
\begin{align}
S^1&:=\frac{S^++S^-}{2}\quad\mbox{and}\quad
S^2:=\frac{S^+-S^-}{2i}\:.
\end{align}
The last spin matrix $S^3$ is already diagonal with respect to $\{\delta_j\}_{j=-J}^J$:
\begin{equation}
S^3\delta_j=j\cdot\delta_j\:.
\end{equation}
\subsection{Definition of the Hamiltonian on finite graphs}
Let $\cG=(\cV,\cE)$ be a finite graph with vertex set $\cV$ and edge set $\cE$. Moreover, assume once and for all that any considered graph $\cG$ is connected. We construct the $(M+1)^{|\mathcal{V}|}$--dimensional\footnote{Here and in the following, $|\mathcal{A}|$ denotes the cardinality of a set $\mathcal{A}$.} Hilbert space $\cH_\cG$ as
\begin{equation}
\cH_\cG:=\bigotimes_{x\in\cV}\cH_x\:,
\end{equation}
where $\cH_x=\C^{M+1}$ is interpreted as a single site Hilbert space located at the site $x$. Let $\Delta>M$. Then, for any $x,y\in\cV$, we define the two--site Hamiltonian $h_{xy}$ as
\begin{equation} \label{eq:twosite}
h_{xy}:=J^2-\frac{1}{\Delta}\left(S_x^1S_y^1+S_x^2S_y^2\right)-S_x^3S_y^3=J^2-\frac{1}{2\Delta}\left(S_x^+S_y^-+S_x^-S_y^+\right)-S_x^3S_y^3\:,
\end{equation}
where we adapt the convention that for any operator $A$ on $\C^{M+1}$ and any $x_0\in\mathcal{V}$, we define $A_{x_0}$ as an operator on $\cH_\cG$ by letting it act as $A$ on $\cH_{x_0}$ and as the identity on all other sites in $\cV\setminus\{x_0\}$. As we shall see below (Remark \ref{rem:nonnegativity}), for $\Delta>M$, we have $h_{xy}\geq 0$. The full XXZ--Hamiltonian $H_\cG$ is then defined by
\begin{equation} \label{eq:finiteXXZ}
H_\cG:=\sum_{x,y\in\cV: x\sim y}h_{xy}\:,
\end{equation}
where $x\sim y$ means that $\{x,y\}\in\cE$. 

\subsection{Conservation of particle number}
Let us define the \emph{local particle number operator} $\cN^{loc}:=(J-S^3)$. Trivially, the spectrum of $\cN^{loc}$ is given by $\sigma(\cN^{loc}):=\{0,1,2,\dots,M\}$ with corresponding eigenvectors 
\begin{equation}
\cN^{loc}\delta_j=(J-j)\delta_j\:.
\end{equation}
For $x\in\cV$, we interpret the eigenvalues $\{J-j\}_{j=-J}^J$ of $\cN^{loc}_x$ as the \emph{number of particles} at the site $x$. In this sense, the number $M=2J$ denotes the \emph{maximal number} of particles that are allowed to occupy any site $x\in\cV$. For later convenience, we relabel the elements of $\{\delta_j\}_{j=-J}^J$ as follows
\begin{equation}
e_k:=\delta_{J-k}\:,
\end{equation}
which defines the orthonormal basis $\{e_k\}_{k=0}^M$. 
The \emph{total particle number operator} $\cN_\cG$ is then just given by
\begin{equation}
\cN_\cG:=\sum_{x\in\cV}\cN_x^{loc}\:,
\end{equation}
with spectrum $\sigma(\cN_\cG)=\{0,1,\dots, |\cV|\cdot M-1, |\cV|\cdot M\}$. Let $\psi_{\bf{0}}:=\bigotimes_{x\in\cV}(e_0)_x$ denote the (non--degenerate) eigenvector of $\cN_\cG$ corresponding to the eigenvalue zero, which we will refer to as the ``vacuum vector". At this point, it will be convenient to introduce the creation operator $a$ and the annihilation operator $a^*$, which are given by
\begin{align}
	a e_k&:= \begin{cases} e_{k+1}\quad &k\in\{0,1,\dots,M-1\}\\0\quad &k=M\end{cases}\\
	a^* e_k&:= \begin{cases} e_{k-1}\quad &k\in\{1,2,\dots,M\}\\0\quad &k=0\end{cases} \:.
\end{align}
Let $\mathfrak{M}$ be the set of all functions from $\cV$ to $\{0,1,2,\dots,M\}$:
\begin{equation} \label{eq:M}
\mathfrak{M}=\left\{\mathfrak{m}:\cV\rightarrow\{0,1,2,\dots,M\}\right\}\:.
\end{equation}
	For any $\mathfrak{m}\in\mathfrak{M}$, define
\begin{equation}
\psi_\mathfrak{m}:=\prod_{x\in\mathcal{V}}(a_x)^{\mathfrak{m}(x)}\psi_{{\bf{0}}}
\end{equation}
with the interpretation that $\psi_\mathfrak{m}$ corresponds to a state with exactly $\mathfrak{m}(x)$ particles located at each site $x$.

For any eigenvalue $N\in\sigma(\cN_\cG)$ let $\mathcal{H}_\cG^N$ denote the corresponding eigenspace (the $N$-particle subspace). It is now important to notice that these eigenspaces reduce the operator $H_\cG$, i.e. the total particle number is conserved under the action of $H_\cG$. We therefore decompose 
\begin{equation}
H_\cG=\bigoplus_{N=0}^{|\cV|\cdot M}H_\cG^N\:,
\end{equation}
where $H_\cG^N:=H_\cG\upharpoonright_{\mathcal{H}_\cG^N}$.
Defining
\begin{equation} \label{eq:MN}
\mathfrak{M}_N=\left\{\mathfrak{m}\in\mathfrak{M}:\sum_{x\in\cV}\mathfrak{m}(x)=N\right\}\:,
\end{equation}
note that we get
\begin{equation}
\mathcal{H}_\cG^N={\spann\{\psi_\mathfrak{m}:\mathfrak{m}\in\mathfrak{M}_N\}}\:.
\end{equation}
\subsection{Construction of the $N$--particle graphs with maximal local occupation number $M$}
For the following construction, note that we are \emph{not} assuming that the underlying graph $\cG=(\cV,\cE)$ has a finite vertex set.

Firstly, we need to slightly modify the definition of $\mathfrak{M}$ given in \eqref{eq:M} to be the set of all finitely supported functions from $\cV$ to $\{0,1,2,\dots,M\}$:
\begin{equation}
\mathfrak{M}=\left\{{\mathfrak{m}}: \cV\rightarrow \{0,1,\dots,M\}, |\supp(\mathfrak{m})|<\infty\right\}\:,
\end{equation}
which is countable. As in \eqref{eq:MN}, we also define the following subsets of $\mathfrak{M}$:
\begin{equation}
\mathfrak{M}_N=\left\{\mathfrak{m}\in\mathfrak{M}:\sum_{x\in\cV}\mathfrak{m}(x)=N\right\}
\end{equation}
and note that
\begin{equation}
\mathfrak{M}=\biguplus_{N\geq 0}\mathfrak{M}_N\:.
\end{equation}
For any $M\in\N$ and any $N\in\{0,1,\dots,|\cV|\cdot M\}$, we now introduce the graph $\cG^{M,N}=(\cV^{M,N},\cE^{M,N})$, which we call ``\emph{$N$--particle graphs with maximal local occupation number $M$}". Its vertex set $\cV^{M,N}$ is given by
\begin{equation}
\cV^{M,N}:=\mathfrak{M}_N
\end{equation}
and the edge set $\cE^{M,N}$ is defined as follows
\begin{equation}
\cE^{M,N}=\left\{\{\mathfrak{m,n}\}\subset\cV^{M,N}: \exists \{x_+,x_-\}\in\cE: \mathfrak{m}(x_\pm)=\mathfrak{n}(x_\pm)\pm 1 \mbox{ and } \forall z\in\cV\setminus\{x_+,x_-\}: \mathfrak{m}(z)=\mathfrak{n}(z) \right\}\:.
\end{equation} 
In other words, two $N$-particle configurations $\mathfrak{m,n}\in\mathfrak{M}_N$ are adjacent to each other in $\cG^{M,N}$ if one configuration can be obtained from the other by moving exactly one particle from configuration $\mathfrak{m}$ along an edge $\{x,y\}\in\cE$ from $x$ to $y$. 
\begin{remark}\label{rem:infinite}
	Note that if $\cG$ is a countable infinite connected graph of bounded maximal degree, so will be $\cG^{M,N}$.
\end{remark}
\begin{remark}
	For the case $M=1$, note that the so--constructed graphs $\cG^{1,N}=(\cV^{1,N},\cE^{1,N}) $ can be identified with the well--known $N$-th symmetric product $\cG^N=(\cV^N,\cE^N)$ of $\cG$ (cf.\ e.g.\ \cite{Carballosa, FS18, O2, O1, Rivera}), where $$\cV^N=\{X\subset \cG:|X|=N\} \mbox{ and } \cE^N=\{\{X,Y\}: X\triangle Y\in\cE\}\:.$$ Since $$\cV^{1,N}=\left\{\mathfrak{m}:\cV\rightarrow \{0,1\}:\sum_{x\in\cV}\mathfrak{m}(x)=N\right\}\:,$$
	where we observe that $\sum_{x\in\cV}\mathfrak{m}(x)=|\mbox{supp}(\mathfrak{m})|$, we identify any $X\in\cV^N$ with the function $\mathfrak{m}_X\in\cV^{1,N}$ given by \begin{equation*} \mathfrak{m}_X(x)=\begin{cases} 1\quad\mbox{if}\quad x\in X\\
	0\quad\mbox{if}\quad x\notin X\:.
	\end{cases}
	\end{equation*}
	From direct inspection, it can be seen that $\{X,Y\}\in\cE^{N}$ if and only if $\{\mathfrak{m}_X,\mathfrak{m}_Y\}\in\cE^{1,N}$.
\end{remark}
\begin{remark} In \cite{HammackSmith}, the reduced $k$--th power $\cG^{(k)}$ of a graph $\cG$ was constructed, which coincides with our construction in the case $N=M=k$. In other words, the $k$--particle graph with maximal local occupation number $k$ is equal to the reduced $k$-th power of $\cG$, i.e.\ $\cG^{k,k}=\cG^{(k)}$.
\end{remark}
\subsection{The XXZ Hamiltonian and its equivalence to a many-particle Schr\"odinger operator}
In the following, we are going to show that the operators $H_{\cG}^N$ are unitarily equivalent to discrete Schr\"odinger--type operators on a suitable Hilbert space. To this end, let $\cG=(\cV,\cE)$ be a finite graph with bounded maximal degree. 

Consider now the Hilbert space $\ell^2(\mathfrak{M}_N)$ equipped with the inner product $\langle \cdot,\cdot\rangle$, given by
\begin{equation}
\langle f,g\rangle=\sum_{\mathfrak{m}\in\mathfrak{M}_N}\overline{f(\mathfrak{m})}g(\mathfrak{m})\:.
\end{equation}
Moreover, let $\phi_\mathfrak{m}\in\ell^2(\mathfrak{M}_N)$ denote the canonical basis vectors given by
\begin{equation} \label{eq:canbasis}
\phi_\mathfrak{m}(\mathfrak{n})=\begin{cases}
1\quad \mbox{if} \quad\mathfrak{m}=\mathfrak{n}\\
0 \quad \mbox{if} \quad\mathfrak{m}\neq\mathfrak{n}\:.
\end{cases}
\end{equation}
Let $\mathfrak{m,n}$ now be two arbitrary configurations in $\cV^{M,N}$ such that $\mathfrak{m}\sim\mathfrak{n}$. In particular, this means that there exists two uniquely determined $x_+(\mathfrak{m,n}),x_-(\mathfrak{m,n})\in{\cV}$ with $\{x_+,x_-\}\in\cE$, such that $\mathfrak{m}(x_\pm)=\mathfrak{n}(x_\pm)\pm 1$ and $\mathfrak{m}(z)=\mathfrak{n}(z)$ for any other $z\in\cV\setminus\{x_+,x_-\}$. (For notational convenience, we omit writing the dependence of $x_\pm$ on $\mathfrak{m}$ and $\mathfrak{n}$.) We then define the \emph{weight function} $w(\mathfrak{m,n})$ to be given by
\begin{align}
	w(\mathfrak{m,n}):&=\left(\frac{M}{2}\left(\mathfrak{m}(x_+)+\mathfrak{n}(x_+)+1\right)-\mathfrak{m}(x_+)\mathfrak{n}(x_+)\right)^{1/2}\cdot\left(\frac{M}{2}\left(\mathfrak{m}(x_-)+\mathfrak{n}(x_-)+1\right)-\mathfrak{m}(x_-)\mathfrak{n}(x_-)\right)^{1/2}\\&=\prod_{x:\mathfrak{m}(x)\neq\mathfrak{n}(x)}\left(\frac{M}{2}\left(\mathfrak{m}(x)+\mathfrak{n}(x)+1\right)-\mathfrak{m}(x)\mathfrak{n}(x)\right)^{1/2}\:.
\end{align}
From this, it is easily observed that $w(\mathfrak{m,n})=w(\mathfrak{n,m})$.
We then define the \emph{weighted adjacency operator} $A_{N}$ on $\cG^{M,N}$ to be given by
\begin{align}
	A_{N}:\qquad\qquad\ell^2(\mathfrak{M}_N)&\rightarrow\ell^2(\mathfrak{M}_N)\\
	(A_{N}f)(\mathfrak{m})&=\sum_{\mathfrak{n}:\mathfrak{n}\sim\mathfrak{m}}w(\mathfrak{m,n})f(\mathfrak{n})\:.
\end{align}
Also, we define the \emph{weighted Laplacian} $\mathcal{L}_{N}$ as follows:
\begin{equation}
(\mathcal{L}_{N}f)(\mathfrak{m})=\sum_{\mathfrak{n}:\mathfrak{n}\sim\mathfrak{m}}w(\mathfrak{m,n})\left(f(\mathfrak{n})-f(\mathfrak{m})\right)\:,
\end{equation}
and note that $-\cL_N$ is a non-negative operator.
The relation between $A_{N}$ and $\mathcal{L}_{N}$ is given by
\begin{equation} \label{eq:laplace}
-\mathcal{L}_{N}=D_N-A_N\:,
\end{equation}
where $D_N$ is the multiplication operator by the \emph{weighted degree function}:
\begin{equation} \label{eq:D}
(D_Nf)(\mathfrak{m})=\left(\sum_{\mathfrak{n}:\mathfrak{n}\sim\mathfrak{m}}w(\mathfrak{m,n})\right)f(\mathfrak{m})\:.
\end{equation}
Moreover, for any configuration $\mathfrak{m}$, we define the \emph{total interaction potential} $V_N$ to be the multiplication operator given by
\begin{align}
	V_N:\qquad\ell^2(\mathfrak{M}_N)&\rightarrow\ell^2(\mathfrak{M}_N)\\ \label{eq:V}
	(V_Nf)(\mathfrak{m})&:=\left(\sum_{\{x,y\}\in\cE}v(\mathfrak{m}(x),\mathfrak{m}(y))\right)f(\mathfrak{m})\:,
\end{align}
where the value of the \emph{two--site interaction potential} $v$ is given by
\begin{align}
	v(\mathfrak{m}(x),\mathfrak{m}(y))=\frac{M}{2}(\mathfrak{m}(x)+\mathfrak{m}(y))-\mathfrak{m}(x)\mathfrak{m}(y)\:.
\end{align}
\begin{remark} \label{rem:infiniteprep}
	Again, note that for the definition of $A_N, D_N, \cL_N$ and $V_N$ it is not necessary to assume that $\cG$ is a finite graph. As mentioned in Remark \ref{rem:infinite}, as long as $\cG$ is of bounded maximal degree, so will be $\cG^{M,N}$, which implies that $A_N, D_N$ and $\cL_N$ are well--defined. Moreover, note that $v(0,0)=0$. Now, since for any $\mathfrak{m}\in\cV^{M,N}$, we have $\sum_{x\in\cV}\mathfrak{m}(x)=N$, this implies that $\mathfrak{m}$ is non--zero for only finitely many $x\in\cV$. Hence, the sum on the right--hand side of \eqref{eq:V} is only over finitely many terms.  
\end{remark}
We are now prepared to state the equivalence between the spin-$J$ XXZ model on any finite graph and many--particle Schr\"odinger operators:
\begin{proposition} \label{prop:equivalence}
	Let $2J=M\in\mathbb{N}$ and assume that $\cG=(\cV,\cE)$ is a finite connected graph. For any $N\in\mathcal\{1,\dots, M\cdot|\cV|\}$, let the unitary map $U_N$ be given by
	\begin{align}
		U_N:\qquad \cH_\cG^N&\rightarrow\ell^2(\mathfrak{M}_N)\\
		\psi_\mathfrak{m}&\mapsto \phi_\mathfrak{m}\:.
	\end{align}
	Then, 
	\begin{equation}
	U_NH_\cG^NU_N^*=-\frac{1}{2\Delta}A_N+V_N=-\frac{1}{2\Delta}\cL_N+V_N-\frac{1}{2\Delta}D_N\:.
	\end{equation}
	
\end{proposition}
\begin{proof} For any $\mathfrak{m}\in\mathfrak{M}_N$, write $\psi_\mathfrak{m}$ as $\psi_\mathfrak{m}=\bigotimes_{x\in \cV}(e_{\mathfrak{m}(x)})_x$. For any $x,y\in\cV$ we get
	\begin{align}
		&\left(J^2-S^3_xS^3_y\right)(e_{\mathfrak{m}(x)})_x\otimes(e_{\mathfrak{m}(y)})_y=\left(J^2-(J-\mathfrak{m}(x))\cdot(J-\mathfrak{m}(y))\right)(e_{\mathfrak{m}(x)})_x\otimes(e_{\mathfrak{m}(y)})_y\\=&\left(J\left(\mathfrak{m}(x)+\mathfrak{m}(y)\right)
		-\mathfrak{m}(x)\mathfrak{m}(y)\right)(e_{\mathfrak{m}(x)})_x\otimes(e_{\mathfrak{m}(y)})_y=v(\mathfrak{m}(x),\mathfrak{m}(y))(e_{\mathfrak{m}(x)})_x\otimes(e_{\mathfrak{m}(y)})_y	\end{align}
	Consequently, we get
	\begin{equation}
	\left(\sum_{\{x,y\}\in\cE}(J^2-S_x^3S_y^3)\right)\psi_\mathfrak{m}=\left(\sum_{\{x,y\}\in\cE}v(\mathfrak{m}(x),\mathfrak{m}(y))\right)\psi_\mathfrak{m}\:,
	\end{equation}
	from which 
	\begin{equation} \label{eq:potential}
	U_N\left(\sum_{\{x,y\}\in\cE}(J^2-S_x^3S_y^3)\right)U_N^*=V_N
	\end{equation}
	follows.
	
	Next, let us consider
	\begin{equation}
	S^+_xS^-_y(e_{\mathfrak{m}(x)})_x\otimes(e_{\mathfrak{m}(y)})_y\:,
	\end{equation}
	where we have to distinguish two cases:
	\begin{itemize}
		\item Case I: $\mathfrak{m}(x)\neq 0$ and $\mathfrak{m}(y)\neq M$. In this case, by a direct calculation, it can be verified that
		\begin{align}
			&S^+_xS^-_y(e_{\mathfrak{m}(x)})_x\otimes(e_{\mathfrak{m}(y)})_y\\=&\left(M\mathfrak{m}(x)-(\mathfrak{m}(x)-1)\mathfrak{m}(x)\right)^{1/2}\cdot\left(M(\mathfrak{m}(y)+1)-(\mathfrak{m}(y)+1)\mathfrak{m}(y)\right)^{1/2}(e_{\mathfrak{m}(x)-1})_x\otimes(e_{\mathfrak{m}(y)+1})_y\:.\label{eq:spinstuff}
		\end{align}
		Now, define the function $\mathfrak{n}_{xy}\in\mathfrak{M}_N$ via $\mathfrak{n}_{xy}(x):=\mathfrak{m}(x)-1$, $\mathfrak{n}_{xy}(y):=\mathfrak{m}(y)+1$, and $\mathfrak{n}_{xy}(z):=\mathfrak{m}(z)$ for any $z\in\cV\setminus\{x,y\}$. With this, we can rewrite \eqref{eq:spinstuff} and get
		\begin{align}
			&S^+_xS^-_y(e_{\mathfrak{m}(x)})_x\otimes(e_{\mathfrak{m}(y)})_y\\
			=&\left(\frac{M}{2}\left(\mathfrak{m}(x)+\mathfrak{n}_{xy}(x)+1\right)-\mathfrak{m}(x)\mathfrak{n}_{xy}(x)\right)^{1/2}\left(\frac{M}{2}\left(\mathfrak{m}(y)+\mathfrak{n}_{xy}(y)+1\right)-\mathfrak{m}(y)\mathfrak{n}_{xy}(y)\right)^{1/2}(e_{\mathfrak{n}_{xy}(x)})_x\otimes(e_{\mathfrak{n}_{xy}(y)})_y\:.
		\end{align}
		This shows that for any $\{x,y\}\in\cE$ and for any $\mathfrak{m}\in\mathfrak{M}_N$ such that $\mathfrak{m}(x)\neq M$ and $\mathfrak{m}(y)\neq 0$, we get
		\begin{equation}
		S^+_xS^-_y\psi_\mathfrak{m}=w(\mathfrak{m},\mathfrak{n}_{xy})\cdot\psi_{\mathfrak{n}_{xy}}\:.
		\end{equation}
		\item Case II: $\mathfrak{m}(x)=0$ or $\mathfrak{m}(y)=M$. In this case, it can easily be seen that
		\begin{equation}
		S^+_xS^-_y(e_{\mathfrak{m}(x)})_x\otimes(e_{\mathfrak{m}(y)})_y=0\:.
		\end{equation}
	\end{itemize}
	
	Defining
	\begin{equation}
	\mathcal{O}(\mathfrak{m}):=\{(x,y)\in\cV\times\cV: \{x,y\}\in\cE, \mathfrak{m}(x)\neq 0 \mbox{ and } \mathfrak{m}(y)\neq M\}\:,
	\end{equation}
	we obtain
	\begin{align} \label{eq:edges}
		\sum_{\{x,y\}\in\cE}(S_x^+S_y^-+S_x^-S_y^+)\psi_\mathfrak{m}=\sum_{(x,y)\in\mathcal{O}(\mathfrak{m})} w(\mathfrak{m},\mathfrak{n}_{xy})\psi_{\mathfrak{n}_{xy}}
	\end{align}
	Now, observe that for any $\mathfrak{m}\in\cV^{M,N}$ we have
	\begin{equation}
	\{\mathfrak{n}\in\cV^{M,N}:\mathfrak{n}\sim\mathfrak{m}\}=\{\mathfrak{n}_{xy}: (x,y)\in\mathcal{O}(\mathfrak{m})\}
	\end{equation}
	and thus the calculation in \eqref{eq:edges} can be continued to get
	\begin{equation}
	\sum_{\{x,y\}\in\cE}(S_x^+S_y^-+S_x^-S_y^+)\psi_\mathfrak{m}=\sum_{(x,y)\in\mathcal{O}(\mathfrak{m})} w(\mathfrak{m},\mathfrak{n}_{xy})\psi_{\mathfrak{n}_{xy}}=\sum_{\mathfrak{n}:\mathfrak{n}\sim\mathfrak{m}}w(\mathfrak{m,n})\psi_\mathfrak{n}\:.
	\end{equation}
	This implies that
	\begin{equation} \label{eq:kinetic}
	U_N\left(\sum_{\{x,y\}\in\cE}(S_x^+S_y^-+S_x^-S_y^+)\right)U_N^*=A_N\:.
	\end{equation}
	Collecting \eqref{eq:potential} and \eqref{eq:kinetic}, we therefore conclude
	\begin{equation}
	U_NH_\cG^NU_N^*=-\frac{1}{2\Delta}A_N+V_N\overset{\eqref{eq:laplace}}{=}-\frac{1}{2\Delta}\cL_N+V_N-\frac{1}{2\Delta}D_N\:,
	\end{equation}
	which is the desired result.
\end{proof}
For later purposes, the following lemma will be useful:
\begin{lemma} \label{lemma:relbound1}
	We have the following estimates
	\begin{equation} \label{eq:estimate}
	-\frac{2J}{\Delta}\left(J^2-S_x^3S_y^3\right)\leq -\frac{1}{2\Delta}\left(S_x^+S_y^-+S_x^-S_y^+\right)\leq \frac{2J}{\Delta}\left(J^2-S_x^3S_y^3\right)
	\end{equation}
	and 
	\begin{equation} \label{eq:estimate2}
	\left(1-\frac{2J}{\Delta}\right)\left(J^2-S_x^3S_y^3\right)\leq h_{xy}\leq \left(1+\frac{2J}{\Delta}\right)\left(J^2-S_x^3S_y^3\right)\leq \frac{M}{2}\left(1+\frac{M}{\Delta}\right)\left(\mathcal{N}^{loc}_x+\mathcal{N}^{loc}_y\right)\:.
	\end{equation}
\end{lemma}
\begin{proof}
	The proof of \eqref{eq:estimate} can be found in \cite[Proof of Thm.\ 5.3]{MNSS}. The first two inequalities in \eqref{eq:estimate2} are an immediate consequence of \eqref{eq:estimate}, while the third inequality can be seen from 
	\begin{equation} \label{eq:esttwosite}
	J^2-S_x^3S_y^3=J^2-(J-\mathcal{N}^{loc}_x)(J-\mathcal{N}^{loc}_y)=J\left(\mathcal{N}_x^{loc}+\mathcal{N}_y^{loc}\right)-\mathcal{N}_x^{loc}\mathcal{N}_y^{loc}\leq J\left(\mathcal{N}_x^{loc}+\mathcal{N}_y^{loc}\right)\:.
	\end{equation}
	(Recall that $M=2J$.)
\end{proof}
\begin{remark} \label{rem:nonnegativity}
	Observe that the operator $(J^2-S_x^3S_y^3)$ is non--negative with 
	\begin{equation}
	\ker(J^2-S_x^3S_y^3)=\mbox{span}\{e_0\otimes e_0, e_M\otimes e_M\}\:.
	\end{equation}
	 Then, from \eqref{eq:esttwosite}, the lower bound
	\begin{equation} \label{eq:boundbelow}
	\left(1-\frac{2J}{\Delta}\right)\left(J^2-S_x^3S_y^3\right)\leq h_{xy}\:,
	\end{equation}
	implies that if $\Delta>M$, the operator $h_{xy}$ is non--negative.\footnote{Note that we are not claiming that $\Delta>M$ is necessary for $h_{xy}$ to be non--negative.} By an explicit calculation, it is not hard to see that $\mbox{span}\{e_0\otimes e_0, e_M\otimes e_M\}\subset\ker(h_{xy})$ and thus -- since by \eqref{eq:boundbelow} we have $\ker(h_{xy})\subset\ker(J^2-S_x^3S_y^3)$, this actually shows that $\ker(h_{xy})=\mbox{span}\{e_0\otimes e_0, e_M\otimes e_M\}$.
\end{remark}
\begin{remark}
Of course, it is also possible to add an external magnetic field $W$ of the form
\begin{equation}
W=\sum_{x\in \cV}W_x\cN^{loc}_x
\end{equation}
to $H_\cG^N$, with $W_x$ being the strength of the magnetic field at the site $x$. It is not hard to see that for any $f\in\ell^2(\mathfrak{M}_N)$:
\begin{equation}
\Big(U_NWU_N^*f\Big)(\mathfrak{m})=\left(\sum_{x\in\cV}W_x\mathfrak{m}(x)\right)f(\mathfrak{m})\:.
\end{equation}
\end{remark}
\subsection{Definition of the Hamiltonian on infinite graphs}
Following \cite{FS18}, let us now introduce the spin--$J$ XXZ Hamiltonian for infinite graphs. To this end, we introduce the Hilbert space
\begin{equation}
\cH_\cG=\ell^2(\mathfrak{M})=\left\{f:\mathfrak{M}\rightarrow\C, \sum_{\mathfrak{m}\in\mathfrak{M}}|f(\mathfrak{m})|^2<\infty\right\}\:,
	\end{equation}
equipped with the inner product
\begin{equation}
\langle f,g\rangle=\sum_{\mathfrak{m}\in\mathfrak{M}}\overline{f(\mathfrak{m})}g(\mathfrak{m})\:.
\end{equation}
Its canonical basis $\left\{\phi_\mathfrak{m}\right\}_{\mathfrak{m}\in\mathfrak{M}}$ is defined as in \eqref{eq:canbasis}.
Likewise, we define the $N$--particle subspaces $\cH_\cG^N:=\ell^2(\mathfrak{M}_N)$.
By Remark \ref{rem:infiniteprep}, the operators $H_\cG^N$ given by
\begin{equation}
H_\cG^N=-\frac{1}{2\Delta}A_N+V_N
\end{equation}
on $\cH_\cG^N=\ell^2(\mathfrak{M}_N)$ are well--defined.
The full XXZ--Hamiltonian $H_\cG$ on $\cH_\cG$ is then given by taking the direct sum over $H_\cG^N$:
\begin{equation}
H_\cG:=\bigoplus_{N\geq 0}H_\cG^N\quad\mbox{with domain}\quad \ell^2(\mathfrak{M})=\bigoplus_{N\geq 0}\ell^2(\mathfrak{M}_N)\:.
\end{equation}
In view of Proposition \ref{prop:equivalence}, the operators $H_\cG^N$ are related to the formal infinite sum $\sum_{x\sim y}h_{xy}$ in the following way: We interpret any basis vector $\phi_\mathfrak{m}\in\ell^2(\mathfrak{M}_N)$ as a configuration $\psi_\mathfrak{m}$ of $N$ particles with each site $x\in\cV$ being occupied by $\mathfrak{m}(x)$ of these $N$ particles. Since we have the constraint $\sum_{x\in\cV}\mathfrak{m}(x)=N$, this implies $\mathfrak{m}(x)\neq 0$ for only finitely many sites $x$. Since -- as an operator on $\C^{M+1}\otimes\C^{M+1}$ -- we have $h_{xy}(e_0)_x\otimes(e_0)_y=0$, this implies that $\sum_{x\sim y}h_{xy}$ contributes only finitely many non--zero terms when acting on $\psi_\mathfrak{m}$. This allows us to extend the meaning of $\sum_{x\sim y}h_{xy}$ to finite linear combinations of basis vectors: 
To this end, for any $N\in\{0,1,2,\dots\}$, let $\ell_0^2(\mathfrak{M}_N)$, be the space of finite linear combinations of canonical $N$--particle configurations:
\begin{equation}
\ell_0^2(\mathfrak{M}_N):=\mbox{span}\{\psi_\mathfrak{m}:\mathfrak{m}\in\mathfrak{M}_N\}\:.
\end{equation} 
Now, for any $\phi\in\ell_0^2(\mathfrak{M})$ we get $\cN_x^{loc}\phi\neq 0$ for only finitely many $x\in\cV$, and thus the total particle number operator $\cN$ -- formally given by $\cN=\sum_{x\in\cV}\cN_x^{loc}$ -- is also well--defined on $\ell_0^2(\mathfrak{M}_N)$ and yields
\begin{equation}
\cN\phi=\sum_{x\in\cV}\cN_{x}^{loc}\phi=N\phi\:.
\end{equation}
By \eqref{eq:estimate2}, for any element $\phi\in\ell^2_0(\mathfrak{M}_N)$, we get
\begin{equation}
\langle \phi,h_{xy}\phi\rangle\leq \frac{M}{2}\left(1+\frac{M}{\Delta}\right)\left\langle \phi,\left(\cN_x^{loc}+\cN_y^{loc}\right)\phi\right\rangle
\end{equation}
and therefore
\begin{equation}
0\leq\langle \phi,H_\cG^N\phi\rangle=\left\langle\phi,\sum_{x\sim y}h_{xy}\phi\right\rangle\leq \frac{M}{2}\left(1+\frac{M}{\Delta}\right)\left\langle\phi,\sum_{x\sim y}(\cN_x^{loc}+\cN_y^{loc})\phi\right\rangle\leq \frac{MNd_{max}}{2}\left(1+\frac{M}{\Delta}\right)\|\phi\|^2\:,
\end{equation}
where $d_{max}$ is the maximal degree of $\cG$. This shows that $H_\cG^N$ -- so far only defined on elements of the dense set $\ell_0^2(\mathfrak{M}_N)$ -- has a unique bounded self--adjoint extension to the entire Hilbert space $\ell^2(\mathfrak{M}_N)$, which we will denote by the same symbol and which satisfies
\begin{equation}
\left\|H_\cG^N\right\|\leq \frac{MNd_{max}}{2}\left(1+\frac{M}{\Delta}\right)\:.
\end{equation}
The full XXZ Hamiltonian $H_\cG$ on $\cG$ is then the self--adjoint operator defined as the direct sum over all particle numbers
\begin{equation}
H_\cG=\bigoplus_{N\geq 0}H_\cG^N\quad\mbox{with domain}\quad \ell^2(\mathfrak{M})=\bigoplus_{N\geq 0}\ell^2(\mathfrak{M}_N)\:.
\end{equation}

Note that if $\cG$ is an infinite graph, the resulting operator $H_\cG$ will be unbounded:
\begin{proposition}
	If $\cG$ is an infinite graph, then $H_\cG$ is unbounded.
\end{proposition}
\begin{proof}
	
This can be easily seen from the following argument: for any $N$, take two sets $\{x_i\}_{i=1}^N, \{y_i\}_{i=1}^N$ such that $\{x_i,y_i\}\in\cE$ for any $i$ such that for any $i\neq j$, we have $\{x_i,y_i\}\cap\{x_j,y_j\}=\emptyset$. Since the two--site operators $h_{xy}$ are non--negative, we then get
\begin{equation} \label{eq:lowbound}
H_{00}^N:=\sum_{i=1}^Nh_{x_iy_i}\leq \sum_{\{x,y\}\in\cE}h_{xy}=H_\cG^N\:.
\end{equation}
\begin{remark}
As before, this statement has to be understood only on $\ell^2_0(\mathfrak{M}_N)$ and can then be extended by continuity.
\end{remark}
Now, observe that 
\begin{equation}
\langle (e_0)_x\otimes (e_1)_y,h_{xy}(e_0)_x\otimes (e_1)_y\rangle=\frac{M}{2}\:.
\end{equation}
We now take $\widehat{\mathfrak{m}}\in\mathfrak{M}_N$ given by
\begin{equation}
\widehat{\mathfrak{m}}(x)=\begin{cases} 1\quad&\mbox{if}\quad x=x_i\\ 0 &\mbox{else}\end{cases}
\end{equation}
and finish the proof by observing that
\begin{equation} 
\langle \psi_{\widehat{\mathfrak{m}}}, H_{00}^N\psi_{\widehat{\mathfrak{m}}}\rangle=\frac{NM}{2}\:,
\end{equation}
which by \eqref{eq:lowbound} is a lower bound for $\|H_\cG^N\|$.
\end{proof}
For later purposes, the following lemma will also be useful:
\begin{lemma} \label{lemma:relbound2}
For any $M,N$, we have the following estimate:
\begin{equation} \label{eq:firstbound}
A_N\leq 2MV_N\:.
\end{equation}
Moreover, let $\cV'\subset\mathcal{V}^{M,N}$ and let $\cG'=(\cV',\cE')$ be the subgraph of $\cG^{M,N}$ induced by $\cV'$, i.e. $\cE'=\{\{\mathfrak{m,n}\}\in\cE^{M,N}: \mathfrak{m,n}\in\cV'\}$. Letting $A_N'$ be the weighted adjacency operator on $\cG'$, we get
\begin{equation} \label{eq:secondbound}
A_N'\leq 2MV_N\upharpoonright_{\ell^2(\cV')}\:.
\end{equation}
\end{lemma}
\begin{proof}
From Lemma \ref{lemma:relbound1}, Estimate \eqref{eq:estimate}, we have for any $x\sim y$:
\begin{equation}
S_x^+S_y^-+S_x^-S_y^+\leq 4J(J^2-S_x^3S_y^3)\:.
\end{equation}
Summing over all $\{x,y\}\in\cE$ together with \eqref{eq:potential} and \eqref{eq:kinetic} yields \eqref{eq:firstbound}. The second bound follows from the fact that if $\phi\in\ell^2(\cV')$, we get $\langle \phi,A_N'\phi\rangle=\langle \phi,A_N\phi\rangle$. Then, from \eqref{eq:firstbound}, we have
\begin{equation}
\langle \phi,A_N'\phi\rangle=\langle \phi,A_N\phi\rangle\leq \langle \phi,2MV_N\phi\rangle\:,
\end{equation}
which finishes the proof.
\end{proof}

\section{An abstract result on the existence of spectral gaps} \label{sec:gaps}
Now, let $G=(V,E)$ be a countably infinite graph of bounded maximal degree. Moreover, let $A$ be a weighted adjacency operator on $\ell^2(V)$ given by
\begin{equation} \label{eq:wgraphadj}
(Af)(x)=\sum_{y\in V: y\sim x}w(x,y)f(y)\:,
\end{equation}
where the weight function $w(x,y)$ is assumed to be symmetric, i.e.\ $w(x,y)=w(y,x)$ and uniformly bounded from above and below, i.e. there exist constants $C_1,C_2>0$ such that 
\begin{equation} \label{eq:weightbound}
\forall x,y\in V:\quad C_1\leq w(x,y)\leq C_2\:.
\end{equation}
We now want to consider Schr\"odinger--type operators $H$ that are of the form
\begin{equation}
H=-gA+U\:,
\end{equation}
where $g>0$ is a coupling constant and $U$ is a multiplication by a real--valued  potential. In addition, we assume this potential to exhibit a spectral gap $\gamma:=U_2-U_1>0$, where $(U_1,U_2)\in\rho(U)$ -- the resolvent set of $U$.

Our goal is to find a sufficient condition on $g$ which ensures the existence of gaps in the spectrum of $H$ which is better than just a norm estimate involving $g\|A\|$. We will closely follow the approach from \cite[Sec.\ 3]{FS18}, where a result of this form has been shown for non--weighted Laplacians on general graphs.

To this end, let $V_1$ and $V_2$ now be the two disjoint subsets of $V$ with $V_1\cup V_2=V$, which are chosen such that
\begin{equation} \label{eq:potsupp}
U_1=\sup_{x\in V_1}U(x)\qquad\mbox{and}\qquad U_2=\inf_{y\in V_2}U(y)\:.
\end{equation}
Now, define the \emph{boundary hopping operator} $B$ from $\ell^2(V_2)$ to $\ell^2(V_1)$ to be given by
\begin{equation}
(Bf)(x):=\sum_{y\in V_2: y\sim x} w(x,y)f(y))
\end{equation}
for any $f\in\ell^2(V_2)$ and any $x\in V_1$. Likewise, its adjoint $B^*$ maps from $\ell^2(V_1)$ to $\ell^2(V_2)$ as follows:
\begin{equation}
(B^*g)(y)=\sum_{x\in V_1: x\sim y}w(y,x)g(x)
\end{equation}
for any $f\in\ell^2(V_1)$ and any $y\in V_2$.
Let us now estimate the norms of $B$ and $B^*$:
\begin{lemma}
Let 
\begin{equation}
d_1:=\sup_{x\in V_1}\left\{\sum_{y\in V_2: y\sim x}w(x,y)\right\}\quad\mbox{and}\quad d_2:=\sup_{y\in V_2}\left\{\sum_{x\in V_1: x\sim y}w(x,y)\right\}\:.
\end{equation}
We then get that $B$ and $B^*$ are bounded with
\begin{equation}
\|B\|=\|B^*\|\leq \sqrt{d_1d_2}\:.
\end{equation}
\end{lemma}
\begin{proof}
Firstly, note that since $G$ is assumed to be of bounded maximal degree and by virtue of \eqref{eq:weightbound}, we get that $d_1$ and $d_2$ are both finite.

Now, for any $f\in \ell^2(V_2)$, consider
\begin{align}
\|Bf\|^2&=\sum_{x\in V_1}|Bf(x)|^2=\sum_{x\in V_1}\left|\sum_{y\in V_2: y\sim x} w(x,y)f(y)\right|^2\\&\leq \sum_{x\in V_1}\left[\left(\sum_{y\in  V_2: y\sim x} w(x,y)\right)\left(\sum_{y\in  V_2: y\sim x}w(x,y)|f(y)|^2\right)\right]\leq d_1\sum_{x\in V_1}\sum_{y\in  V_2: y\sim x}w(x,y)|f(y)|^2\:,
\end{align}
where the first estimate follows from a simple application of the Cauchy-Schwarz inequality. Exchanging the order of summation then yields:
\begin{equation}
d_1\sum_{x\in V_1}\sum_{y\in  V_2: y\sim x}w(x,y)|f(y)|^2=d_1\sum_{y\in V_2}\sum_{x\in  V_1: x\sim y}w(x,y)|f(y)|^2\leq d_1d_2\sum_{y\in V_2}|f(y)|^2=d_1d_2\|f\|^2\:,
\end{equation}
from which the lemma follows. 
\end{proof}
Next, we want to present the main result of this section. Given the decomposition of $V$ into two disjoint sets $V_1$ and $V_2$ such that $V=V_1\cup V_2$ and \eqref{eq:potsupp} is satisfied, we define the restricted graphs $G_i=(V_i,E_i)$, where $E_i=\{\{x,y\}\in E: x,y\in V_i\}$ for $i\in\{1,2\}$. With $A_i$, we denote the weighted adjacency operators on these restricted graphs.

\begin{proposition} \label{prop:gapcond}
Let $G=(V,E)$ be a countable graph of bounded maximal degree and for any $g>0$, let $H=-gA+U$ be a Schr\"odinger--type operator, where $A$ is given by \eqref{eq:wgraphadj} and its weights $w(x,y)$ satisfy \eqref{eq:weightbound}. Moreover, assume the potential $U$ exhibits a spectral gap $(U_1,U_2)\subset\rho(U)$ and decompose $V=V_1\cup V_2$, where $V_1\cap V_2=\emptyset$ such that \eqref{eq:potsupp} is satisfied. Moreover, let $C$ be any multiplication operator that bounds $A_2$ from above, i.e.\ satisfies $A_2\leq C$.
If 
\begin{equation} \label{eq:apriori}
E\in \left( \sup_{x\in V_1} U(x)+g\|A_1\|, \inf_{x\in V_2}(U(x)-gC(x))\right)
\end{equation}
satisfies in addition the condition
\begin{equation} \label{eq:gapcondition}
g^2<\frac{\inf_{x\in V_2}(U(x)-gC(x)-E)\cdot (E-\sup_{x\in V_1}U(x)-g\|A_1\|)}{d_1d_2}\:,
\end{equation}
 then $E\in\rho(H)$. Moreover, the Hilbert space dimension of the spectral projection of $H$ onto $(-\infty,E)$ is at least $|V_1|$.
\end{proposition}
\begin{proof} The proof of this proposition is completely analogous to the proof for the non-weighted case, which can be found in \cite[Prop.\ 3.2]{FS18}.
\end{proof}
\begin{remark}
Note that since the Laplacian on $V_2$, given by $-\cL_2=D_2-A_2$, where $D_2$ is the weighted degree function on the subgraph $G_2$, is non-negative, we get
\begin{equation}
0\leq -\cL_2=D_2-A_2\quad\Rightarrow\quad D_2\geq A_2\:.
\end{equation}
Moreover, note that for any $x\in V_2$, we get
\begin{equation}
D_2(x)=\sum_{y\in V_2: x\sim y}w(x,y)\leq \sum_{y\in V_2: x\sim y}w(x,y)+\sum_{y\in V_1: x\sim y}w(x,y)=D(x)
\end{equation}
and thus $A_2\leq D_2\leq D\upharpoonright_{\ell^2(V_2)}$, where $D$ is the weighted degree function on the entire graph, we can always choose $C=D$ in Condition \eqref{eq:gapcondition}.
\end{remark}
\section{Gaps in the low--lying spectrum} \label{sec:droplet}
Let us now apply the results of the previous sections, in order to find sufficient conditions on $\Delta$ that ensure the existence of a spectral gap above what we will call the ``droplet spectrum". To this end, we will consider two types of graphs -- the chain and the strip of width $L$. As the spin--$1/2$ case (corresponding to $M=1$) has been studied in \cite{FS18}, we will assume $M\in\{2,3,\dots\}$ from now on.   
\subsection{Ground state gap and droplet spectrum}
For an infinite graph $\cG=(\cV,\cE)$, consider again the operators
\begin{equation}
H_\cG^N=-\frac{1}{2\Delta} A_N+V_N\:.
\end{equation}
Firstly, note that for $N=0$, the Hilbert space $\cH_0^\cG=\ell^2(\mathfrak{M}_0)$ is one--dimensional and $H_\cG^0$ acts as the zero--operator on $\ell^2(\mathfrak{M}_0)$. From this, we get $\sigma(H_\cG^0)=\{0\}$, which is the ground state energy. For any other $N\in\N$, we use Lemma \ref{lemma:relbound2}, in order to estimate
\begin{equation} \label{eq:gsgap}
H_\cG^N=-\frac{1}{2\Delta}A_N+V_N\geq \left(1-\frac{M}{\Delta}\right)V_N\:.
\end{equation}
Now, observe that for any $\mathfrak{m}\in\mathfrak{M}_N$, we have
\begin{equation}
V_N(\mathfrak{m})=\sum_{\{x,y\}\in\cE}v(\mathfrak{m}(x),\mathfrak{m}(y))\:,
\end{equation}
where $v(\mathfrak{m}(x),\mathfrak{m}(y))=0$ if and only if $\mathfrak{m}(x)=\mathfrak{m}(y)=0$ or $\mathfrak{m}(x)=\mathfrak{m}(y)=M$ and $v(\mathfrak{m}(x),\mathfrak{m}(y))\geq M/2$ else. Now, since $\cG$ is an infinite graph, for any $\mathfrak{m}\in\mathfrak{M}_N$, due to the requirement that $\sum_x\mathfrak{m}(x)=N$, note that there must exist at least one edge $\{x_0,y_0\}\in\cE$ such that $\mathfrak{m}(x_0)\neq\mathfrak{m}(y_0)$ and consequently $v(\mathfrak{m}(x_0),\mathfrak{m}(y_0))\geq M/2$. Consequently, for any $N\in\N$, we get $V_N\geq M/2$ and therefore we continue the estimate in \eqref{eq:gsgap} to get
\begin{equation}
H_\cG^N\geq\left(1-\frac{M}{\Delta}\right)V_N\geq \left(1-\frac{M}{\Delta}\right)\frac{M}{2}\:.
\end{equation}
This means that for infinite graphs, for $\Delta>M$, there is a gap of width at least $(1-\frac{M}{\Delta})\frac{M}{2}$ above the ground state energy $0$ in $\sigma(H_\cG)$.
Now, let $V_{N,1}$ and $V_{N,2}$ denote the lowest and second lowest eigenvalue of $V_N$, respectively. By $\mathcal{V}^{M,N}_1$ we then denote the set of configurations $\mathfrak{m}\in\mathcal{V}^{M,N}$ for which $V_N$ attains its minimum:
\begin{equation}
\mathcal{V}^{M,N}_1:=\left\{\mathfrak{m}\in\mathcal{V}^{M,N}:V_N(\mathfrak{m})=V_{N,1}\right\}\:.
\end{equation}
Motivated by the spin-$1/2$ case (cf.~\cite{FS18}), we call the set $\mathcal{V}_{1}^{M,N}$ the set of \emph{classical droplet configurations}. Moreover, we define $\mathcal{V}_2^{M,N}:=\mathcal{V}^{M,N}\setminus\mathcal{V}_1^{M,N}$. Let $P_1$ be the orthogonal projection onto the space of classical droplet configurations $\ell^2(\mathcal{V}_1^{M,N})$ and $P_2:=\idty-P_1$. From \eqref{eq:gsgap}, we get
\begin{equation}
P_2H_\cG^NP_2\geq \left(1-\frac{M}{\Delta}\right)V_NP_2\geq \left(1-\frac{M}{\Delta}\right)V_{N,2}P_2\:.
\end{equation}
This motivates us to define the \emph{droplet spectrum} $\delta_N$ of $H_\cG^N$ as
\begin{equation}
\delta_N:=\sigma\left(H_\cG^N\right)\cap\left(0,\left(1-\frac{M}{\Delta}\right)V_{N,2}\right)\:.
\end{equation}
In the following two sections, we will now use Proposition \ref{prop:gapcond} in order to find sufficient conditions on $\Delta$ that ensure the existence of a spectral gap above the droplet spectrum for the case of the chain and the strip.
\subsection{The chain}
Let us now consider the chain, i.e.\ $\cG=(\cV,\cE)$, where $\cV=\Z$ and $\{x,y\}\in\cE :\Leftrightarrow |x-y|=1$. 
We are now interested in finding the configurations $\mathfrak{m}$ for which $V_N(\mathfrak{m})$ attains its minimum. We will only provide a sketch of the proof for the following proposition, as it is very technical and not particularly insightful.

\begin{proposition} 
Let $\cG$ be the chain. Moreover, let $N= kM$, where $k\in\{2,3,\dots\}$. Then, the minimizers $\{\mathfrak{m}_j\}_{j=1}^M$ of $V_N$ are up to translations given by: 
	\begin{equation} \label{eq:minconfig}
	\mathfrak{m}_j(x)=\begin{cases} j\quad&\mbox{if}\quad x=0\\
	                                M\quad&\mbox{if}\quad x\in\{1,2,\dots,k-1\}\\
	                                (M-j)\quad&\mbox{if}\quad x=k\\
	                                0\quad&\mbox{else.}
	                                \end{cases}
	\end{equation}
	In particular, we get $V_N(\mathfrak{m}_j)=M^2$ for any $j\in\{1,2,\dots,M\}$.	
\end{proposition}
\begin{proof}[Sketch of proof]
	Firstly, observe that for any $\mathfrak{m}\in\cV^{M,N}$, the interaction potential $V_N(\mathfrak{m})$ can be simplified to
	\begin{equation} 
	V_N(\mathfrak{m})=\sum_{x\in\Z}\frac{M}{2}(\mathfrak{m}(x)+\mathfrak{m}(x+1))-\sum_{x\in\Z}\mathfrak{m}(x)\mathfrak{m}(x+1)=NM-\sum_{x\in\Z}\mathfrak{m}(x)\mathfrak{m}(x+1)\:,
	\end{equation}
	where we have used that $\sum_x\mathfrak{m}(x)=N$. In other words, finding minimizers of $V_N$ is equivalent to finding maximizers of 
	\begin{equation} \label{eq:productsum}
		\sum_{x\in\Z}\mathfrak{m}(x)\mathfrak{m}(x+1)\:.
	\end{equation}
Now, if a configuration $\mathfrak{m}$ is not of the form \eqref{eq:minconfig}, it will not maximize \eqref{eq:productsum} as there are two mechanisms that could increase the value of \eqref{eq:productsum}: firstly a rearrangement of the values of $\mathfrak{m}$, such that larger values are neighbored (increasing up to a ``centered" maximal value $\mathfrak{m}(x_{max})$ and decreasing from there) and secondly, a stacking of particles towards that centered maximum occupation number, such that larger numbers get multiplied with each other in \eqref{eq:productsum}.
\end{proof}
Let us illustrate this argument by an example:
\begin{example}\label{ex:minimizer}
	 Let $M=3$ and $N=3M=9$. Let $\mathfrak{m}(1)=3, \mathfrak{m}(2)=0, \mathfrak{m}(3)=1, \mathfrak{m}(4)=1, \mathfrak{m}(5)=0, \mathfrak{m}(6)=2, \mathfrak{m}(7)=2$ and $\mathfrak{m}(x)=0$ else. For convenience, we represent the values of this function by an array of the form $(3,0,1,1,0,2,2)$. Now, firstly we rearrange the array such that larger numbers get multiplied with each other: $(0,1,1,3,2,2,0)$, which is increasing towards the maximal value $3$ and decreasing from there. Now, we move the particles towards the center and obtain $(0,0,2,3,3,1,0)$ corresponding to a translation of the configuration $\mathfrak{m}_2$ as described in \eqref{eq:minconfig}.
\end{example}
Let us now apply Proposition \ref{prop:gapcond} to the Hamiltonians $H_\cG^{kM}$, where $k\in\{2,3,\dots\}$ in order to show the existence of a spectral gap for the restricted XXZ Hamiltonian
\begin{equation} \label{eq:restricted}
\widehat{H}:=\bigoplus_{k\geq 2} H_\cG^{kM}\:.
\end{equation}
The main reason for choosing to only treat this restricted Hamiltonian is for the purpose of simplicity of presentation and in order to avoid many additional special cases that have to be considered when allowing the full XXZ Hamiltonian.
For an application of Proposition \ref{prop:gapcond}, we now split $\cV^{M,N}$ into $V_1$ and $V_2$, where $V_1$ is the set of all configurations $\mathfrak{m}\in\cV^{M,N}$, for which $V_N$ attains its minimal value $M^2$, while $V_2$ is its complement. Writing $N=kM$, let us now collect the necessary quantities:
\begin{itemize}
	\item The coupling constant $g=\frac{1}{2\Delta}$.
	\item The potential $U$ is given by $V_N$ as defined in \eqref{eq:V}. In particular, by choice of $V_1$, we have $\sup_{\mathfrak{m}\in V_1}V_N(\mathfrak{m})=M^2$.
	\item Observe that none of the configurations $\mathfrak{m}_j$ or their translates are adjacent to each other in $\cG^{M,N}$. Thus, we get $\|A_1\|=0$.
	\item By Lemma \ref{lemma:relbound2}, we have $A_1\leq 2MV_N\upharpoonright_{\ell^2(V_1)}$. Thus, we choose $C:=2MV_N\upharpoonright_{\ell^2(V_1)}$ and therefore
	\begin{equation} \label{eq:potchoice}
	U-gC=\left(1-\frac{M}{\Delta}\right)V_N\:.
	\end{equation}
	As we have assumed that $\Delta>M$, this ensures that the expression in \eqref{eq:potchoice} is non--negative.
	Since $V_2$ is the set of all configurations for which $V_N$ does not attain its minimal value $M^2$ and since $V_N$ is integer--valued, we get 
	\begin{equation} \label{eq:V2}
	\inf_{\mathfrak{m}\in V_2}V_N(\mathfrak{m})\geq M^2+1\:.
	\end{equation}
	(In fact, it is possible to find a configuration proving that equality holds in \eqref{eq:V2}. Expressed as an array of numbers as in Example \ref{ex:minimizer}, one such configuration is given by $(M,M,\dots,M,M-1,1)$) Therefore, by \eqref{eq:potchoice} and \eqref{eq:V2}, we get
	\begin{equation}
	\inf_{\mathfrak{m}\in V_2} \Big(U(\mathfrak{m})-gC(\mathfrak{m})\Big)\geq \left(1-\frac{M}{\Delta}\right)(M^2+1)\:.
	\end{equation} 
	\item Finally, we need to estimate the numbers $d_1$ and $d_2$. In order to estimate $d_1$, we begin with the observation that for any element of $V_1$, which has to be of the form $\mathfrak{m}_j$ as given by \eqref{eq:minconfig}, there are either two adjacent configurations ($j=M$) or four adjacent configurations ($1\leq j\leq M-1$). As in Example \ref{ex:minimizer}, we represent them by arrays of occupation numbers. Then the configuration $\mathfrak{m}_M$, represented by $(0,M,M,\dots,M,M,0)$ is adjacent to the configurations 
	\begin{align} \label{eq:twoconfig}
	&(1,M-1,M,\dots,M,M,0)\\&(0,M,M,\dots,M,M-1,1)\notag\:.
	\end{align}
	 For $1\leq j\leq M-1$, the configuration $\mathfrak{m}_j$ -- represented by $(0,j,M,M,\dots,M,M,M-j,0)$ is adjacent to the following four configurations 
	\begin{align} \label{eq:fourconfig}
	&(1,j-1,M,M,\dots,M,M,M-j,0)\\ &(0,j+1,M-1,M,\dots,M,M,M-j,0)\notag\\ &(0,j,M,M,\dots,M,M-1,M-j+1,0)\notag\\ &(0,j,M,M,\dots,M,M,M-j-1,1)\notag\:.
	\end{align}
	From the explicit expression for the weight function $w$, it is then not hard to show that
	\begin{equation}
	d_1=\sup_{j\in\{1,\dots,M\}}\left(\sum_{\mathfrak{n}\in V_2:\mathfrak{n}\sim\mathfrak{m}_j}w(\mathfrak{m}_j,\mathfrak{n})\right)\leq 4M^{3/2}\:.
	\end{equation}
	Conversely, in order to estimate $d_2$, note that the only configurations in $V_2$ that are adjacent to configurations in $V_1$ are -- up to translations -- given by the configurations in \eqref{eq:twoconfig} and \eqref{eq:fourconfig}. Each of the two configurations in \eqref{eq:twoconfig} would be adjacent to $\mathfrak{m}_M$, while for a fixed $j$, the four configurations in \eqref{eq:fourconfig} would be adjacent to $\mathfrak{m}_j$. Again, from the explicit expression for $w$, we therefore find
	\begin{equation}
	d_2\leq M^{3/2}\:.
	\end{equation}
\end{itemize}
Let us now use these observations in order to apply Proposition \ref{prop:gapcond} to find energy values $E$ that are not in the spectrum of $H_\cG^N$. We begin with the \emph{a priori} Condition \eqref{eq:apriori}, which requires that 
\begin{equation}
E\in \left( \sup_{x\in V_1} U(x)+g\|A_1\|, \inf_{x\in V_2}(U(x)-gC(x))\right)=\left(M^2,\left(1-\frac{M}{\Delta}\right)(M^2+1)\right)\:.
\end{equation}  
In order for this to make sense, we need to require
\begin{equation}
M^2< \left(1-\frac{M}{\Delta}\right)(M^2+1) \quad \Leftrightarrow\quad M^3+M<\Delta\:.
\end{equation}
Now, plugging everything into Condition \eqref{eq:gapcondition}, we find that if $E$ satisfies in addition the condition
\begin{equation} \label{eq:quadratic}
\frac{1}{4\Delta^2}<\frac{\Big(\left(1-\frac{M}{\Delta}\right)(M^2+1)-E\Big)\cdot(E-M^2)}{4M^3}
\end{equation}
we have $E\in\rho(H_\cG^N)$.
By further analyzing the quadratic equation in $E$ coming from \eqref{eq:quadratic}, one finds that if 
\begin{equation}
\Delta> M^3+2M^{3/2}+M\:,
\end{equation}
then this implies that the operator $H_\cG^N$ has a spectral gap that -- up to an error of $C(M)/\Delta^2$ -- coincides with
\begin{equation}
\left(M^2,\left(1-\frac{M}{\Delta}\right)(M^2+1)\right)\:.
\end{equation}
Note that since this estimate is uniform in particle numbers $N$ that are of the form $N=Mk$, this also shows the existence of the same gap in the spectrum of the restricted XXZ Hamiltonian $\widehat{H}$, which we defined in \eqref{eq:restricted}.
\subsection{The strip of width $L$}
Let $L\in\{2,3,\dots\}$ and consider the strip of width $L$, i.e. $\cG=(\cV,\cE)$, where $\cV=\Z\times\{1,\dots,L\}$ and $(z_1,\ell_1)\sim(z_2,\ell_2):\Leftrightarrow |z_1-z_2|+|\ell_1-\ell_2|=1$. Again, for technical convenience, we will only consider a special subset particle numbers $N$. To be more specific, we define
\begin{equation}
\mathcal{S}=\{kLM: k\in\N, k\geq L/2\}
\end{equation}
and consider the restricted Hamiltonian
\begin{equation} \label{eq:restricted2}
\widehat{H}:=\bigoplus_{N\in\mathcal{S}}H_\cG^N\:.
\end{equation}
 This ensures that -- up to translations in $\Z$--direction -- the unique configuration $\widehat{\mathfrak{m}}$ that minimizes $V_N$ is a rectangular configuration given by
\begin{equation}
\widehat{\mathfrak{m}}(z,\ell)=\begin{cases} M\quad &\mbox{if}\quad z\in\{1,2,\dots,k\}\\ 0 &\mbox{else,}\end{cases}
\end{equation}
with $V_N(\widehat{\mathfrak{m}})=LM^2$. Again, we define the set of configurations that minimize $V_N$ to be $V_1$ with $V_2$ being its complement. We proceed with collecting the necessary quantities which we need in order to apply Proposition \ref{prop:gapcond}:
\begin{itemize}
	\item The coupling constant $g=\frac{1}{2\Delta}$.
	\item The potential $U$ is given by $V_N$ as defined in \eqref{eq:V}. In particular, by choice of $V_1$, we have $\sup_{\mathfrak{m}\in V_1}V_N(\mathfrak{m})=V_N(\widehat{\mathfrak{m}})=LM^2$.
	\item Neither the configuration $\widehat{\mathfrak{m}}$ nor its translates are adjacent to each other in $\cG^{M,N}$. Thus, we get $\|A_1\|=0$.
	\item By Lemma \ref{lemma:relbound2}, we have $A_1\leq 2MV_N\upharpoonright_{\ell^2(V_1)}$. Thus, we choose $C:=2MV_N\upharpoonright_{\ell^2(V_1)}$ and therefore
	\begin{equation} \label{eq:potchoice2}
	U-gC=\left(1-\frac{M}{\Delta}\right)V_N\:.
	\end{equation}
	Since we made the assumption that $\Delta>M$, note that the expression in \eqref{eq:potchoice2} is non--negative.
	Since $V_2$ is the set of all configurations for which $V_N$ does not attain its minimal value $M^2$ and since $V_N$ is integer--valued, we get 
	\begin{equation} \label{eq:V3}
	\inf_{\mathfrak{m}\in V_2}V_N(\mathfrak{m})\geq LM^2+1\:.
	\end{equation}
	(We actually believe that the correct lower bound for $V_N\upharpoonright_{V_2}$ is given by $V_N\upharpoonright_{V_2}\geq LM^2+M$, however we do not have a rigorous proof for this at this point.)
	Therefore, by \eqref{eq:potchoice2} and \eqref{eq:V3}, we get
	\begin{equation}
	\inf_{\mathfrak{m}\in V_2} \Big(U(\mathfrak{m})-gC(\mathfrak{m})\Big)\geq \left(1-\frac{M}{\Delta}\right)(LM^2+1)\:.
	\end{equation} 
	\item Finally, we need to determine the numbers $d_1$ and $d_2$. In order to find $d_1$, we begin with the observation that there are exactly $2L$ configurations that are adjacent to $\widehat{m}$ -- each coming from breaking off one particle at each from the $2L$ boundary sites of $\mbox{supp}(\widehat{\mathfrak{m}})$ and moving them to the adjacent unoccupied site. 
    For any such adjacent configuration $\mathfrak{n}$, we compute
    \begin{equation} \label{eq:w}
    w(\widehat{\mathfrak{m}},\mathfrak{n})=M
    \end{equation}
    and thus we get $d_1=2LM$. Conversely, the only configurations in $V_2$ that are adjacent to $\widehat{\mathfrak{m}}$ or its translates or those configurations $\mathfrak{n}$ obtained above from breaking single particles off the boundary sites from configuration $\widehat{\mathfrak{m}}$. Each of these configurations $\mathfrak{n}$ is adjacent to precisely one configuration in $V_1$ and using \eqref{eq:w}, we get $d_2=M$. 
\end{itemize}
Again, let us now apply Proposition \ref{prop:gapcond} and find energy values $E$ that are not in the spectrum of $H_\cG^N$. From the \emph{a priori} Condition \eqref{eq:apriori}, which requires that 
\begin{equation}
E\in \left( \sup_{x\in V_1} U(x)+g\|A_1\|, \inf_{x\in V_2}(U(x)-gC(x))\right)=\left(LM^2,\left(1-\frac{M}{\Delta}\right)(LM^2+1)\right)\:,
\end{equation}  
we get the condition
\begin{equation}
LM^2< \left(1-\frac{M}{\Delta}\right)(LM^2+1) \quad \Leftrightarrow\quad LM^3+M<\Delta\:.
\end{equation}
Hence, by Proposition \ref{prop:gapcond}, Condition \eqref{eq:gapcondition}, we conclude that if $E$ satisfies in addition the condition
\begin{equation} \label{eq:quadratic2}
\frac{1}{4\Delta^2}<\frac{\Big(\left(1-\frac{M}{\Delta}\right)(LM^2+1)-E\Big)\cdot(E-LM^2)}{2LM^2}\:,
\end{equation}
we have $E\in\rho(H_\cG^N)$.
A further analysis of the quadratic equation in $E$ coming from \eqref{eq:quadratic2}, shows that if 
\begin{equation}
\Delta>LM^3+\left(1+\sqrt{2L}\right)M \:,
\end{equation}
then this implies that the operator $H_\cG^N$ has a spectral gap that -- up to an error of $C(L,M)/\Delta^2$ -- coincides with
\begin{equation}
\left(LM^2,\left(1-\frac{M}{\Delta}\right)(LM^2+1)\right)\:.
\end{equation}
Note that since this estimate is uniform in particle numbers $N\in\mathcal{S}$, this also shows the existence of the same gap in the spectrum of the restricted XXZ Hamiltonian $\widehat{H}$, which we defined in \eqref{eq:restricted2}.
\bibliographystyle{amsplain}

\end{document}